\documentclass[11pt]{article}

\usepackage[utf8]{inputenc}

\usepackage{fullpage}
\usepackage{array}
\usepackage{paralist}
\usepackage{verbatim}
\usepackage{algpseudocode}
\usepackage{algorithm}
\usepackage{hyperref}
\usepackage[cmex10]{amsmath}
\usepackage{amsthm}
\usepackage{amssymb}
\usepackage{amsfonts}
\usepackage{graphicx}
\usepackage{epstopdf}
\usepackage{dblfloatfix}
\usepackage{color}
\usepackage{xypic}
\usepackage{relsize}
\usepackage{tikz}
\usetikzlibrary{matrix, calc, arrows,decorations.pathmorphing}
\usepackage{tkz-graph}
\usepackage{stmaryrd}

\theoremstyle{definition}
\newtheorem{example}{Example}
\theoremstyle{definition}
\newtheorem{definition}{Definition}
\theoremstyle{plain}

\newtheorem{lemma}{Lemma}

\newcommand{\rk}{\operatorname{rank}}
\renewcommand{\dim}{\operatorname{dim}}
\newcommand{\F}{\mathbb{F}_2}

\newcommand{\spn}{\operatorname{span}}

\input{Qcircuit}

\title{Polynomial-time T-depth Optimization of Clifford+T circuits via Matroid Partitioning}

\author{
{Matthew Amy,$^{1}$ Dmitri Maslov,$^2$\;\;and Michele Mosca$^{3,4}$} \\
{\small\it $^1$ Department of Computer Science} \\
{\small\it University of Toronto, Toronto, Ontario, Canada} \\
{\small\it $^2$ National Science Foundation} \\
{\small\it Arlington, Virginia, USA} \\
{\small\it $^3$  Institute for Quantum Computing, and Dept. of Combinatorics \& Optimization} \\
{\small\it University of Waterloo,  Waterloo, Ontario, Canada} \\
{\small\it $^4$ Perimeter Insitute for Theoretical Physics} \\
{\small\it Waterloo, Ontario, Canada} \\
}

\begin{document}

\maketitle

\begin{abstract}
Most work in quantum circuit optimization has been performed in isolation from the results of quantum fault-tolerance. Here we present a polynomial-time algorithm for optimizing quantum circuits that takes the actual implementation of fault-tolerant logical gates into consideration. Our algorithm re-synthesizes quantum circuits composed of Clifford group and $T$ gates, the latter being typically the most costly gate in fault-tolerant models, e.g., those based on the Steane or surface codes, with the purpose of minimizing both $T$-count and $T$-depth. A major feature of the algorithm is the ability to re-synthesize circuits with additional ancillae to reduce $T$-depth at effectively no cost. The tested benchmarks show up to 65.7\% reduction in $T$-count and up to 87.6\% reduction in $T$-depth without ancillae, or 99.7\% reduction in $T$-depth using ancillae.
\end{abstract}

\section{Introduction}

Quantum computers have the potential to efficiently solve important computational problems, including integer factorization \cite{s94} and quantum simulation \cite{l96}, for which there are no known efficient classical algorithms. However, even with recent advances in quantum information processing technologies \cite{bskjfubb12, bwcomklw11, cgcmsrpkrrks12, rgppccsmrkrks12}, the prospects of scalable quantum computing without some systematic way of mitigating physical errors and noise are bleak.

The active and vibrant fields of quantum error correction and fault-tolerance provide such tools for constructing scalable quantum computers. By combining physical qubits through the use of error correcting codes and providing fault-tolerant logical operations, larger computations can be achieved with high fidelity -- by concatenating codes, or in topological codes by increasing code distance -- provided the physical operations achieve a certain threshold fidelity. With recent improvements to fault-tolerant thresholds \cite{baokm12, fsg09, fwh12}, scalable quantum computation is becoming more and more viable, resulting in a growing need for efficient automated design tools targeting fault-tolerant quantum computers.

Quantum circuit synthesis and optimization is particularly important, given the prevalence of the circuit model of quantum computation, but previous work has been largely isolated from the unique concerns of fault-tolerance. While at the physical level, coupled gates are generally the hardest to perform, most of the common quantum error-correcting codes have efficient $CNOT$ implementations. Moreover, for fault-tolerant models based on (double even, self-dual) CSS codes, e.g., the popular Steane code, as well as the promising surface codes, the Clifford group can be implemented as logical gates with little cost \cite{fsg09, zlc00}.

For universal quantum computing, however, at least one non-Clifford group gate is needed, which typically requires large ancilla factories and gate teleportation to implement fault-tolerantly. As the non-Clifford $T$ gate has known constructions in most of the common error correction schemes, the standard universal fault-tolerant gate set is taken to be ``Clifford + $T$". Given the high cost of the fault tolerant implementations of the $T$ gate \cite{agp06, fsg09}, exceeding the cost of Clifford group gates by as much as a factor of a hundred or more, it has recently been proposed that efficient circuits should minimize the number of $T$ gates, and more specifically the number of $T$ gates that cannot be performed in parallel \cite{ammr13, f12} -- we define these metrics as a circuit's $T$-count and $T$-depth, respectively. Indeed, Fowler \cite{f12} shows how to perform fault-tolerant computations in time proportional to one round of measurement per layer of $T$ gates, and as a result the $T$-depth directly determines a circuit's runtime. Likewise, reducing the number of $T$ gates reduces the number of ancilla states that require preparation, vastly reducing circuit volume and at the same time increasing the fidelity of the computation. While the primary purpose of our work is to optimize $T$-depth (circuit runtime), our algorithm also provides significant reductions to $T$-count (circuit volume).

Some recent work has been done concerning minimization of $T$-depth \cite{ammr13, s13}, though these previous results focus on finding small optimal two- and three-qubit circuits \cite{ammr13}, and on classes of circuits that can be parallelized to $T$-depth 1 by adding ancillae \cite{ammr13, s13}. By contrast, we report a scalable automated tool for the optimization of $T$-depth that functions with or without ancillae, and is not limited to a few qubits or a specific class of circuits. In particular, we present a polynomial-time algorithm for optimizing both the $T$-depth and $T$-count of quantum circuits composed of Clifford group and $T$ gates. The algorithm also makes automatic use of ancillae to optimize $T$-depth, with the addition of ancillae typically decreasing the runtime of our software implementation. Our experiments show on average 61.1\% reduction in $T$-depth and 39.9\% reduction in $T$-count without adding any ancillae, using the available benchmarks.  When the use of ancillae is allowed, the average $T$-depth reduction is demonstrated to be as high as 80.7\% (the more ancillae are allowed the more parallelization becomes possible, in some cases reducing $T$-depth by as much as 99.7\%).

The rest of this paper is structured as follows: Section \ref{sec:prelim} reviews some background on quantum and reversible computation, and introduces the notations we will use. Sections \ref{sec:cnott} and \ref{sec:mat} describe the algorithmic core -- a procedure that optimally parallelizes the $T$ gates in a circuit composed of $CNOT$ and $T$ gates by performing matroid partitioning. Section \ref{sec:uni} develops a heuristic extending the optimal $\{CNOT, T\}$ core to a universal gate set, while Section \ref{sec:tpar} describes the final algorithm. Section \ref{sec:res} reports our experimental results, and Section \ref{sec:con} concludes the paper.

\section{Preliminaries}
\label{sec:prelim}

We begin by reviewing some basic facts about quantum and reversible circuits necessary for this paper. 

In the classical circuit model, the state of a system of $n$ bits is represented as a binary string of length $n$, with classical gates corresponding to operators that map length-$n$ binary strings to length-$m$ binary strings. More precisely, length-$n$ binary strings are vectors of $\mathbb{F}_2^n$, where $\mathbb{F}_2$ is the two-element finite field with addition corresponding to Boolean exclusive-OR (EXOR, $\oplus$) and multiplication corresponding to Boolean AND ($\land$). We then represent classical gates as operators $f:\mathbb{F}_2^n\to\mathbb{F}_2^m$, and we typically refer to $f$ as a (classical) function. For brevity, if $m=1$ we call $f$ {\it Boolean}.

The quantum circuit model, one of the prominent models of quantum computation \cite{nc00}, generalizes the classical circuit model to deal with quantum effects. In particular, it describes the state space of a system of $n$ qubits as a vector in a $2^n$-dimensional complex vector space $\mathcal{H}$ spanned by the (classical) $n$ bit states. By convention, we refer to the classical states as the standard or computational basis of $\mathcal{H}$ and write them in Dirac notation: $\ket{x}, x\in\F^n$.

In contrast to the classical circuit model, quantum gates are restricted to a subset of all operators on $\mathcal{H}$ -- specifically, quantum gates are linear operators $U:\mathcal{H}\to\mathcal{H}$ that preserve the $L_2$ norm. Such operators $U$ satisfy $U^\dagger U=UU^\dagger=I$, where $U^\dagger$ denotes the adjoint of $U$, and are known as {\it unitary}.  Given that unitary operators are invertible, we see that the subset of quantum transformations that permute the computational basis states are exactly the set of invertible classical transformations -- we call such functions {\it reversible}, with the intuition that any computation performed by reversible functions can be undone or reversed. The {\it Toffoli} gate, $$\Lambda_2(X):|x\rangle|y\rangle|z\rangle\mapsto|x\rangle|y\rangle\ket{z \oplus (x\land y)},$$ is an example of a reversible function.

We can also have classical/quantum computations that use {\it ancillae} -- being bits/qubits that can be initialized to the $0$/$\ket{0}$ or $1$/$\ket{1}$ state and act as a temporary register.  Without loss of generality, we require that all ancillae are initialized in the $0/\ket{0}$ state. In the case of a circuit with $n$ bits, $m$ of which are data bits (i.e. $n-m$ is the number of ancillae), we describe the state space as some subspace $V$ of $\F^n$ with dimension $m$. We will typically use $n$ to represent the total number of bits in a system, and $m$ to refer to the number of data bits.

While all reversible classical gates are linear as operators over $\mathcal{H}$, they need not be linear as operators over $\mathbb{F}_2^n$. In particular, we call $f:\mathbb{F}_2^n\to\mathbb{F}_2^m$ {\it linear} if $f(x\oplus y) = f(x)\oplus f(y)$. For instance, the reversible {\it controlled-NOT} gate $$CNOT:|x\rangle|y\rangle\mapsto|x\rangle|x\oplus y\rangle$$ is linear over $\F^n$. It is a known result that the set of all linear reversible functions are those that can be computed by a circuit consisting of only $CNOT$ gates \cite{pmh08}.

Throughout this paper we will also be interested in linear Boolean functions and their relation to linear reversible functions. For convenience, we refer to the set of $n$-ary linear Boolean functions $\F^n\to\F$ as the {\it dual vector space} $(\F^n)^*$ of $\F^n$, and note that a linear Boolean function $f:\F^n\to\F$ can be represented as a row vector over $\mathbb{F}_2^n$ -- i.e., $x^T$ for some $x\in\mathbb{F}_2^n$. Furthermore, for a set of linear Boolean functions $S\subseteq(\F^n)^*$, we define $\rk(S)$ as the maximum number of independent (row) vectors in $S$, or equivalently the dimension of the subspace $V^*$ spanned by $S$.

As this paper is concerned with the optimization of quantum circuits, we also define some quantum gates commonly used in fault tolerant models. In particular, we define the $T$ and Hadamard gates, $$T:\ket{x}\mapsto e^{\frac{i\pi x}{4}}\ket{x},\quad H:\ket{x}\mapsto\frac{\ket{0}+(-1)^x\ket{1}}{\sqrt{2}}.$$ We will show that circuits over the gate set $\{CNOT, T\}$ implement linear reversible functions with discrete phases corresponding to the eighth roots of unity. A side result of this paper is a proof that $\{CNOT, T\}$ circuits can be simulated on a classical computer in polynomial time. If this set is further extended with the Hadamard gate we achieve a gate set that is universal for quantum computation. We call  $\{H, CNOT, T\}$ the ``Clifford + $T$" gate set, as it contains the Clifford group generators $\{H, P:=T^2, CNOT\}$ along with the $T$ gate. Moreover, $\{H, CNOT, T\}$ is a minimal generating set for the Clifford + $T$ gate set. 

Since quantum gates are commonly defined by unitary matrices, we provide the equivalent matrix definitions below:
$$H:=\frac{1}{\sqrt{2}}\begin{pmatrix} 1 & 1 \\ 1 & -1 \end{pmatrix},\quad CNOT:=\begin{pmatrix} 1 & 0 & 0 & 0 \\ 0 & 1 & 0 & 0 \\ 0 & 0 & 0 & 1 \\ 0 & 0 & 1 & 0 \end{pmatrix},\quad T:=\begin{pmatrix} 1 & 0 \\ 0 & e^{\frac{i\pi}{4}} \end{pmatrix}.$$


\subsection{Computable sets of linear Boolean functions}

While every linear reversible function $f$ over $n$ inputs can be written as $n$ linear Boolean functions, i.e. $$f(a_1,a_2,...,a_n) = \left(f_1(a_1,...,a_n), ..., f_n(a_1,...,a_n)\right),$$ not every set of linear Boolean functions defines a reversible function. For instance, $f(a_1,a_2,a_3) = (a_1, a_2, a_1\oplus a_2)$ is irreversible since the input $a_3$ is effectively destroyed. It is easy to verify that a set of $n$ linear Boolean functions forms the outputs of an $n$-ary linear reversible function if and only if they have rank equal to the dimension of the input space.

\begin{lemma}
\label{lem:reversibility}
Given a subspace $V$ of $\F^n$ and a set of linear Boolean functions $S=\{f_1,f_2\dots f_n\}\subseteq V^*$, the linear function $f:V\to W$ defined as 
$$f(a_1,a_2,\dots, a_n) = (f_1(a_1,\dots, a_n), \dots, f_n(a_1,\dots, a_n))$$ is reversible if and only if $\rk(S)=\dim(V).$
\end{lemma}


Since the unitary quantum circuit model is reversible, a set of linear Boolean functions $S$ can only be computed simultaneously (i.e. there exists a quantum circuit implementing the transformation $\ket{a_1a_2...a_n}\mapsto\ket{b_1b_2...b_n}$ where for each $f\in S$, $f(a_1,a_2,...,a_n)=b_i$ for some $i$) if it defines a reversible function. We call such a set of linear Boolean functions {\it (reversibly) computable} over a particular input space $V$ -- as we will be concerned strictly with reversible computations, we frequently omit the qualifier {\it reversible}.

We may also want compute a set $S$ of linear Boolean functions with a linear reversible function on $n>|S|$ qubits. In this case, since every $n$-ary linear reversible function corresponds to some set of $n$ linear Boolean functions, a linear reversible function computes $S$ if and only if there exists some computable superset $S'$ of $S$. Equivalently from Lemma~\ref{lem:reversibility}, a set $S\subseteq (\F^n)^*$ is (reversibly) computable over a subspace $V$ of $\F^n$ if and only if there exists a superset $S'$ of $S$ with cardinality $n$ such that $\rk(S') = \dim(V)$.

\begin{figure}
\centerline{
\Qcircuit @C=1em @R=.7em {
\lstick{x_1} & \targ & \qw & \targ & \qw & \qw & \rstick{\!\!x_1\oplus x_2\oplus x_3\oplus x_4} \\ 
\lstick{x_2} & \qw & \ctrl{1} & \qw & \targ & \qw & \rstick{\!\!x_2\oplus x_4} \\
\lstick{x_3} & \qw & \targ & \ctrl{-2} & \qw & \qw & \rstick{\!\!x_2\oplus x_3} \\
\lstick{x_4} & \ctrl{-3} & \qw & \qw & \ctrl{-2} & \qw & \rstick{\!\!x_4}
}}
\caption{A circuit computing the functions $x_1\oplus x_2\oplus x_3\oplus x_4,$ $x_2\oplus x_4,$ and $x_2\oplus x_3$.}
\label{fig:comp}
\end{figure}
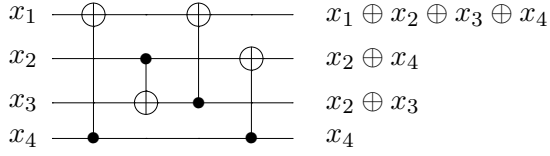

Before moving on, we establish the condition under which a computable superset of linear Boolean functions exists. Given a set $S=\{f_1,f_2,...,f_k\}\subseteq(\F^n)^*$ and subspace $V$, we can observe that we only need to find some $f_{k+1},...,f_n\in(\F^n)^*$ such that $\rk\left(\{f_1,f_2,...,f_n\}\right)=\dim(V)$. It is not hard to see that such $f_{k+1},...,f_n$ exist if and only if the number of linearly independent vectors in $(\F^n)^*$ needed is at most $n-k$.

\begin{lemma}
\label{lem:completebasis}
Given a subspace $V$ of $\F^n$ and a set of linear Boolean functions $S\subseteq (\F^n)^*$, there exists a superset $S'$ of $S$ with cardinality $n$ such that $\rk\left(S'\right)=\dim(V)$ if and only if
\begin{equation}\label{eq:indcond} 
	\dim(V) - \rk\left(S\right) \leq n - |S|.
\end{equation}
\end{lemma}



We can note that inequality~(\ref{eq:indcond}) implies $|S|=\rk(S)$, i.e., $S$ is linearly independent, when $\dim(V)=n$.

\section{\{$CNOT$, $T$\} circuits}
\label{sec:cnott}

We first consider circuits over the gate set $\{CNOT, T, P:=T^2, Z:=T^4, T^\dagger:=T^7, P^\dagger:=T^6\}$, as they have a particular property that will be crucial to synthesizing low $T$-depth circuits.  We remind the reader that the even powers of the $T$ gate are all Clifford gates, whereas all odd powers lie outside the Clifford group.  This is an essential observation for practical considerations.  Furthermore, no power of the $T$ gate requires more than a single non-Clifford $T$ gate to implement.  We usually omit the extraneous gates and refer to this gate set by the generating set $\{CNOT, T\}$. 

It can be observed that since $CNOT\ket{x}\ket{y}=\ket{x}\ket{x\oplus y}$ and $T\ket{x}=e^{\frac{i\pi x}{4}}\ket{x}$ for $x, y\in\F$, a $\{CNOT, T\}$ circuit can be described as computing a linear reversible function on the input basis state, with an added phase that is some power of $\omega:=e^{i\pi / 4}$. Stated more precisely \cite{ammr13},

\begin{lemma}
\label{lem:Tpar}
A unitary $U\in U(2^n)$ is exactly implementable by an $n$-qubit circuit over $\{CNOT, T\}$ if and only if
$$U|x_1x_2\dots x_n\rangle = \omega^{p(x_1,x_2,\dots,x_n)}|g(x_1, x_2, \dots, x_n)\rangle
$$ where $x_1x_2...x_n\in\F^n$ and \vspace{-0.5em}$$p(x_1,x_2,\dots,x_n)=\sum_{i=1}^l c_i\cdot f_i(x_1, x_2, \dots, x_n)$$ for some linear reversible function $g\in\F^n\to\F^n$ and linear Boolean functions $f_1, f_2, ..., f_l\in(\F^n)^*$ with coefficients $c_1, c_2, ..., c_l\in\mathbb{Z}_8$.
\end{lemma}


As a result of Lemma~\ref{lem:Tpar}, we can fully characterize any unitary $U\in U(2^n)$ implementable by a $\{CNOT, T\}$ circuit with a set $S\subseteq\mathbb{Z}_8\times(\F^n)^*$ of linear Boolean functions together with coefficients in $\mathbb{Z}_8$, and a linear reversible output function $g:\F^n\to\F^n$, with the interpretation $$U_{\langle S, g\rangle}\!\!:\!|x_1x_2...x_n\rangle\mapsto\omega^{\mathlarger{\sum}\limits_{(c,f)\in S} \!\!\! c\cdot f(x_1,x_2,\dots,x_n)}|g(x_1, x_2, ..., x_n)\rangle.$$ Moreover, $S$ and $g$ are efficiently computable given a $\{CNOT, T\}$ circuit, taking time linear in the number of qubits and gates. In computing $S$ it also becomes apparent when $T$ gates, possibly physically separated within the circuit, are applied to the same value and can thus be replaced by a phase gate such as $P:\ket{x}\mapsto\omega^{2\cdot x}\ket{x}$. Our experimental results show that a large number of $T$ gates can be removed in this way.


\begin{example}
\label{ex:Trem}

Consider the following circuit.

\centerline{
\Qcircuit @C=.7em @R=.5em {
\lstick{x_1} & \qw & \gate{T_1} & \ctrl{1} & \qw & \targ & \qw & \ctrl{1} & \qw & \qw & \rstick{\!\!\!\!x_2} \\
\lstick{x_2}& \qw & \qw & \targ & \ctrl{1} & \ctrl{-1} & \ctrl{1} & \targ & \gate{T_2^\dagger} & \qw & \rstick{\!\!\!\!x_1} \\
\lstick{x_3} & \ctrl{1} & \qw & \ctrl{1} & \targ & \targ & \targ & \gate{T_3} & \qw & \qw & \rstick{\!\!\!\!x_3\oplus x_4} \\
\lstick{x_4}& \targ & \gate{T_4} & \targ & \qw & \ctrl{-1} & \qw & \qw & \qw & \qw & \rstick{\!\!\!\!x_4}
}
}
\vspace{0.5em}
If we track the effect of the $CNOT$ gates we see that $T_1$ and $T_2^\dagger$ (indices are used to mark different $T$ gates within the circuit) are both applied to qubits in the state $\ket{x_1}$, resulting in a cumulative phase of $\omega^{x_1+7x_1}=\omega^{8x_1}=1^{x_1}=1$. As such, both gates can be removed.  Likewise, $T_3$ and $T_4$ are both applied to qubits in the state $\ket{x_3\oplus x_4}$, and their phases combine to $\omega^{2(x_3\oplus x_4)}$ -- this pair of $T$ gates can thus be replaced with a single $P$ gate. This results in an optimized circuit:
\centerline{
\Qcircuit @C=.7em @R=.5em {
\lstick{x_1} & \qw      & \ctrl{1} & \qw      & \targ     & \qw      & \ctrl{1} & \qw & \rstick{\!\!\!\!x_2} \\
\lstick{x_2} & \qw      & \targ    & \ctrl{1} & \ctrl{-1} & \ctrl{1} & \targ    & \qw & \rstick{\!\!\!\!x_1} \\
\lstick{x_3} & \ctrl{1} & \ctrl{1} & \targ    & \targ     & \targ    & \gate{P} & \qw & \rstick{\!\!\!\!x_3\oplus x_4} \\
\lstick{x_4} & \targ    & \targ    & \qw      & \ctrl{-1} & \qw      & \qw      & \qw & \rstick{\!\!\!\!x_4}
}
}
\vspace{0.5em}

\noindent which may be optimized further to the form $SWAP(x_1,x_2) CNOT(x_4,x_3) P(x_3)$ by rewriting the linear reversible section. 
\end{example}

Once $S$ and $g$ have been computed, the proof of Lemma~\ref{lem:Tpar} \cite{ammr13} gives a constructive method for synthesizing a circuit implementing $U_{\langle S, g\rangle}$. However, this na\"{i}ve method of re-synthesis may end up with {\it worse} $T$-depth than the original circuit, despite possibly reduced $T$-count.

We can instead recall from Section~\ref{sec:prelim} that if a subset $A\subseteq S$ is reversibly computable (i.e. if the linear Boolean functions in $A$ are reversibly computable), we can construct a linear reversible function with outputs simultaneously computing the functions in $A$; in this case $|A|$ of the necessary phase factors could then be applied in parallel. Synthesis can thus proceed by partitioning $S$ into computable subsets, then for each partition $A\subseteq S$ first compute a (reversible) superset of $A$ with a stage of $CNOT$ gates; many efficient algorithms exist \cite{br01, pmh08, m07} that can decompose a linear reversible function into $CNOT$ gates. Next we apply the relevant phase gates in parallel to add the phase $\omega^{\sum_{(c,f)\in A} c\cdot f(x_1,x_2,\dots,x_n)}$, and finally uncompute by reversing the $CNOT$ stage. Given that at most one $T$ gate is used to implement any integer power of $T$, every partition will have a $T$-depth of at most 1.

As a result, any unitary $U$ implementable over $\{CNOT, T\}$ can be implemented in $T$-depth $k$ where $k$ is the minimum number of sets partitioning\footnote{For the remainder of this paper we do not separate $S_0$ and $S_1$ to simplify the presentation, though our algorithm can be trivially modified to partition $S_0$ and $S_1$ separately.} $S_1=\{(c, f)\in S | c\equiv 1 \mod 2\}$ into computable subsets, as elements in $S_0=S\setminus S_1$ do not require $T$ gates to implement. In fact, we can trivially see that given a specific set $S$ and output $g$, $k$ is the minimal $T$-depth, as any layer of $T$ gates in a circuit implementing $U_{\langle S, g\rangle}$ corresponds to a computable subset of $S$. 

It is important to note that while this method reduces and maximally parallelizes the $T$ gates, it will not necessarily find the optimal $T$-count and by extension $T$-depth.  Specifically, given a $\{CNOT, T\}$ circuit with phase defined by the set $S$ there may be some distinct $S'$ that defines an equivalent computation using fewer $T$ gates. Consider, for instance, the set $S_{\emptyset}=\left\{(1, f) | f\in(\F^4)^*\right\}$. We note that the integer sum $\sum_{f\in (\F^4)^*}f(x_1,x_2,x_3,x_4)$ is equal to $8$ for any non-zero $(x_1,x_2,x_3,x_4)\in\F^4$, and $0$ for $(0,0,0,0)$. Since $\omega^8=\omega^0=1$, $S_{\emptyset}$ computes the trivial phase on every input and is therefore equivalent to the empty set $\emptyset$. In this case, all $|S_{\emptyset}|=15$ $T$ gates can then be removed.


It turns out, through a brute force search, that for $n<4$, no two sets $S, S'$ requiring distinct numbers of $T$ gates define equivalent computations. As a result, we obtain a direct proof that the doubly controlled-$Z$ gate requires 7 $T$ gates to implement over $\{CNOT, T\}$ with any number of ancillae. For $n\geq 4$ however, the problem of minimizing $T$-count in $\{CNOT, T\}$ circuits reduces to a minimization problem over degree 1 polynomials in mixed arithmetic; moreover, since every such polynomial defines the global phase for some $\{CNOT, T\}$ circuit, the two problems are in fact equivalent.

\subsection{Parallelizing $\Lambda_2(Z)$}

To illustrate $\{CNOT, T\}$ re-synthesis, consider the circuit in Figure~\ref{fig:ccZ}, implementing the doubly controlled-$Z$ gate, $$\Lambda_2(Z):\ket{x_1x_2x_3}\mapsto\omega^{4\cdot x_1\land x_2\land x_3}\ket{x_1x_2x_3}$$ over $\{CNOT, T\}$. We track the effect of each $CNOT$ gate on the state of the qubits, as annotated in the circuit. When a $T$/$T^\dagger$ gate is applied, we add/subtract a term in the exponent of $\omega$ corresponding to the state of the target qubit. The resulting transformation is then $\Lambda_2(Z):\ket{x_1x_2x_3}\mapsto\omega^{p(x_1,x_2,x_3)}\ket{x_1x_2x_3},$ where \begin{align*}p(x_1,x_2,x_3)=x_1 + x_2 + x_3 - (x_1\oplus x_2) - (x_1\oplus x_3) - (x_2\oplus x_3) + (x_1\oplus x_2\oplus x_3).\end{align*} In fact, since $2\cdot(x\land y) = x + y - x\oplus y$, we see that $\omega^{p(x_1,x_2,x_3)}=\omega^{4\cdot(x_1\land x_2\land x_3)}$ as expected.

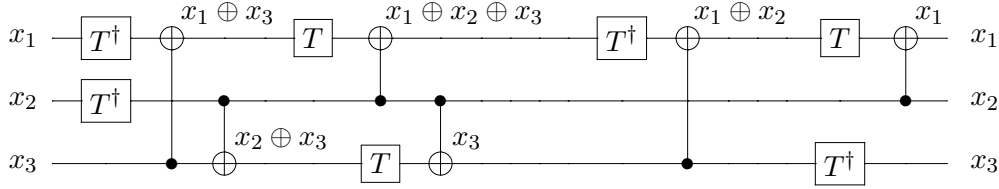
\begin{figure}
\centerline{
\Qcircuit @C=1em @R=.7em {
\lstick{x_1} & \gate{T^\dagger} & \targ & 
\ustick{\:x_1\oplus x_3}\qw & \push{\rule{0em}{1em}}\qw & 
\gate{T} & \targ & 
\ustick{\;\;\;\;\;\;x_1\oplus x_2\oplus x_3} \qw & 
\push{\rule{0em}{1em}}\qw & \push{\rule{0em}{1em}}\qw & \push{\rule{0em}{1em}}\qw & \push{\rule{0em}{1em}}\qw &
\gate{T^\dagger} & \targ & 
\ustick{\;\;\;x_1\oplus x_2} \qw & \push{\rule{0em}{1em}}\qw & \push{\rule{0em}{1em}}\qw &
\gate{T} &
\targ & \ustick{\!\!\!\!\!\!\!x_1}\qw & \rstick{\!\!\!\!x_1} \\
\lstick{x_2} & \gate{T^\dagger} & \qw & 
\ctrl{1} & \push{\rule{0em}{1em}}\qw & 
\qw & \ctrl{-1} & 
\ctrl{1} & \qw & \push{\rule{0em}{1em}}\qw & \push{\rule{0em}{1em}}\qw & \push{\rule{0em}{1em}}\qw &
\qw & \qw & 
\qw & \push{\rule{0em}{1em}}\qw &\push{\rule{0em}{1em}}\qw & \qw &
\ctrl{-1} & \qw & \rstick{\!\!\!\!x_2} \\
\lstick{x_3} & \qw & \ctrl{-2} & 
\targ & \ustick{\:\:\:\:\:x_2\oplus x_3}\qw & 
\push{\rule{0em}{1em}}\qw & \gate{T} & 
\targ & \ustick{\!\!\!\!\!\!x_3}\qw & \push{\rule{0em}{1em}}\qw & \push{\rule{0em}{1em}}\qw & \push{\rule{0em}{1em}}\qw &
\qw & \ctrl{-2} & \qw & \push{\rule{0em}{1em}}\qw &  \push{\rule{0em}{1em}}\qw &
\gate{T^\dagger} & \qw & \qw & \rstick{\!\!\!\!x_3}
}
}
\caption{$\{CNOT, T\}$ circuit implementing the doubly controlled $Z$ gate.}
\label{fig:ccZ}
\end{figure}

The $T$-stages in Figure~\ref{fig:ccZ} also correspond to a partition of $S$, specifically $$\left\{\{-(x_1),-(x_2)\},\{x_1\oplus x_3, x_2\oplus x_3\},\{-(x_1\oplus x_2\oplus x_3)\}, \{x_1\oplus x_2, -(x_3)\}\right\}.$$ It can easily be verified that for each subset $S$ in this partition, $\dim(V) - \rk(S)\leq n-|S|$ (i.e. $S$ satisfies (\ref{eq:indcond})), since $\dim(V)=n=3$ and each subset is linearly independent. By contrast, the partition $$\left\{\{-(x_1),-(x_2)\},\{x_1\oplus x_3, x_2\oplus x_3, x_1\oplus x_2\},\{-(x_1\oplus x_2\oplus x_3)\}, \{-(x_3)\}\right\}$$ could not have been synthesized as the set $\{x_1\oplus x_3, x_2\oplus x_3, x_1\oplus x_2\}$ has rank $2$, thus the mapping $\ket{x_1x_2x_3}\mapsto\ket{(x_1\oplus x_3)(x_2\oplus x_3)(x_1\oplus x_2)}$ is not reversible.

While we haven't yet described how to find a minimal computable partition algorithmically, we can observe that the $T$-depth $3$ partition $$\left\{\{-(x_1),-(x_2), x_1\oplus x_3\}, \{-(x_1\oplus x_2\oplus x_3)\},\{-(x_3), x_1\oplus x_2, x_2\oplus x_3\}\right\}$$ is computable since each subset satisfies (\ref{eq:indcond}), and is also minimal. If, however, we added an ancilla when re-synthesizing Figure~\ref{fig:ccZ}, we can use the extra qubit to generate a smaller partition. In particular, we know $\ket{x_1x_2x_3}\mapsto\ket{x_1x_2(x_1\oplus x_3)}$ is reversible so $\ket{x_1x_2x_3}\ket{0}\mapsto\ket{x_1x_2(x_1\oplus x_3)}\ket{0}$ is as well. We can then compute the value $x_1\oplus x_2\oplus x_3$ into the ancilla with $2$ $CNOT$ gates and apply $4$ $T$ gates (one for each qubit) at the same time. To connect this intuitive idea with equation (\ref{eq:indcond}), we observe that $n=4$, $\dim(V)=3$ since the input $\ket{x_1x_2x_3}\ket{0}$ spans a space of dimension $3$, and $$\rk(\{-(x_1),-(x_2), x_1\oplus x_3, -(x_1\oplus x_2\oplus x_3)\})=3,$$ so $\dim(V) - \rk(S) = 0 \leq 0 = n - |S|.$ Figure~\ref{fig:exsynth} shows the resulting circuit, implementing $\Lambda_2(Z)$ in $T$-depth 2.

\begin{figure}
\centerline{
\Qcircuit @C=.7em @R=.7em {
\lstick{x_1} & \ctrl{2} & \qw & \qw & \gate{T^\dagger} & \qw & \qw & \ctrl{2} & \targ & \qw & \gate{T} & \qw & \targ & \qw & \rstick{x_1} \\
 \lstick{x_2} & \qw & \qw & \ctrl{2} & \gate{T^\dagger} & \ctrl{2} & \qw & \qw & \ctrl{-1} & \targ & \gate{T} & \targ & \ctrl{-1} & \qw & \rstick{x_2} \\
\lstick{x_3} & \targ & \ctrl{1} & \qw & \gate{T} & \qw & \ctrl{1} & \targ & \qw & \ctrl{-1} & \gate{T^\dagger} & \ctrl{-1} & \qw & \qw & \rstick{x_3} \\
\lstick{0} & \qw & \targ & \targ & \gate{T^\dagger} & \targ & \targ & \qw & \qw & \qw & \qw & \qw & \qw & \qw & \rstick{0} \\
}
}
\caption{$T$-depth 2 implementation of Figure~\ref{fig:ccZ} with one ancilla.}
\label{fig:exsynth}
\end{figure}

\section{Matroids}
\label{sec:mat}

We now turn our attention to the problem of determining a minimal partition of a set of linear Boolean functions into computable sets. Due to the connection between the computability condition (\ref{eq:indcond}) on sets of linear Boolean functions and linear independence, we are able to phrase the problem as a {\it matroid partitioning} problem. To do so, we first introduce the concept of a {\it matroid}, an algebraic structure that generalizes the idea of linear independence in vector spaces.

\begin{definition}
\label{def:matroid}
A finite matroid is a pair $(S, I)$ where $S$ is a finite set and $I$ is a set of subsets of $S$ such that
\begin{enumerate}
	\item $\emptyset\in I$.
	\item For all $A, B\subset S$, if $A\in I$ and $B\subset A$, then $B\in I$.
	\item For all $A, B\in I$, if $|A|>|B|$, then there exists some $a\in A$ such that $B\cup\{a\}\in I$.
\end{enumerate}
\end{definition}

It turns out that a set of linear Boolean functions, together with an independence relation defined by the equality~(\ref{eq:indcond}), forms a matroid:

\begin{lemma}
\label{lem:matroid}
For any subspace $V$ of $\mathbb{F}_2^n$ with dimension $m$ and set of linear Boolean functions $S=\{f_1, f_2,\dots, f_k\}\subseteq V^*$, let $I$ denote the set 
$$\{A\subseteq S | m - \rk\left(A\right) \leq n - |A|\}.
$$ 

The pair $(S, I)$ is a finite matroid.
\end{lemma}
\begin{proof}
We verify that $(S, I)$ satisfies all three conditions of Definition~\ref{def:matroid}.
\begin{enumerate}
	\item $m - \rk(\emptyset) \leq n - |\emptyset|$ is trivially true since $m\leq n$. Thus $\emptyset\in I.$
	\item Suppose $A, B\subset S$, where $A\in I$ and $B\subset A$.
		Since $A=B\cup (A\setminus B)$ we see that \begin{align*} \rk(A)&\leq\rk(B)+\rk(A\setminus B) \\ &\leq\rk(B) + (|A| - |B|),\end{align*} and so $\rk(A) - \rk(B) \leq |A| - |B|$. Since $m - \rk(A) \leq n -|A|$, we see that $$m \leq n +\rk(A) -|A|\leq n + \rk(B) - |B|,$$ and thus $m - \rk(B) \leq n -|B|$.

	\item Suppose $A, B\in I$ and $|A|>|B|$. If $\rk(A) \leq \rk(B)$, then $$m-\rk(B)\leq m-\rk(A)\leq n-|A|<n-|B|,$$ and so $m - \rk(B\cup\{s\}) \leq n - |B\cup\{s\}|$ for any $s\in A$. 

		Otherwise, $\rk(A) > \rk(B)$ and we can let $A'$ and $B'$ be maximal linearly independent subsets of $A$ and $B$, respectively. Since $A'\not\subseteq \spn(B')$, for any $s\in A\setminus \spn(B')$, $B'\cup\{s\}$ is linearly independent. Then, \begin{align*}m - \rk(B'\cup\{s\}) &= m - \rk(B\cup\{s\}) \\ &= m - \rk(B) - 1 \\ &\leq n - |B| - 1 \\ &= n - |B\cup\{s\}|.\end{align*}
\end{enumerate}
\end{proof}

With Lemmas \ref{lem:completebasis} and \ref{lem:matroid} we see that the problem of finding a minimal partition of the phase factors in a $\{CNOT, T\}$ circuit reduces to the more general matroid partitioning problem.

\subsection{Matroid partitioning}

The matroid partitioning problem can be defined as follows:
\begin{definition}(Matroid partitioning)

Given a matroid $(S, I)$, find a partition $\{A_1,A_2,\dots, A_k\}$ of $S$ such that $A_i\in I$ for each $1\leq i\leq k$ and for any other partition $\{A'_1, A'_2, \dots, A'_{k'}\}$ into independent subsets, $k'\geq k$.
\end{definition}

Matroid partitioning can, perhaps surprisingly, be solved in polynomial time, given an independence oracle for the matroid \cite{e65}. As a result, given an oracle for (\ref{eq:indcond}), the $T$ gates in a $\{CNOT, T\}$ circuit can be optimally partitioned efficiently. Since the condition in Lemma~\ref{lem:completebasis} can be checked by using Gaussian elimination to compute the matrix rank in $O(n^3)$ time,\footnote{In practice we reduce this to $O(m^2n)$ by storing vectors in $V^*$ as length $\dim(V)=m$ vectors. The $O(n^3)$ bound is used for simplicity.}  we thus see that a minimal partition can be computed in polynomial time. The rest of this section describes an algorithm for computing such a minimal partition.

The algorithm we use for solving the matroid partitioning problem is based on an algorithm due to Edmonds \cite{e65}. Given a matroid $(S, I)$ and a minimal (matroid) partition $P$ of $S'\subset S$, we take an element $s\in S\setminus S'$ not already partitioned and construct a minimal partition of $S'\cup\{s\}$. To create the new partition, we construct a directed graph $G_s$ containing a vertex $u$ for every $u\in S'\cup\{s\}$ as well as a vertex $\bot_p$ for every subset $p\in P$. The edges of $G_s$ represent changes to the partition that are invariant under the property of each subset being independent. In particular, for any $u, v\in S'\cup\{s\}$ there is a directed edge $v\to u$ in $G_s$ if and only if $u$ is contained in some subset $p\in P$ and $(p\setminus \{u\})\cup\{v\}\in I$, i.e. $v$ can be added to $p$ if we remove $u$. Additionally, given a subset $p\in P$ and element $u\in S'\cup\{s\}$, there exists an edge $u\to\bot_p$ if and only if $p\cup\{u\}\in I$. A path from $s$ to $\bot_p$ for some subset $p$ gives a set of updates to $P$ that produce a valid partition $P'$ of $S'\cup\{s\}$. Likewise, if there is no such path, there is no partition of size $|P|$ partitioning $S'\cup\{s\}$ (see \cite{e65} for a proof), and so a new subset $\{s\}$ is added to $P$. Figure~\ref{fig:partition} shows the full graph $G_s$ for one iteration when computing a minimal partition for $\Lambda_2(Z)$ (Figure~\ref{fig:ccZ}).

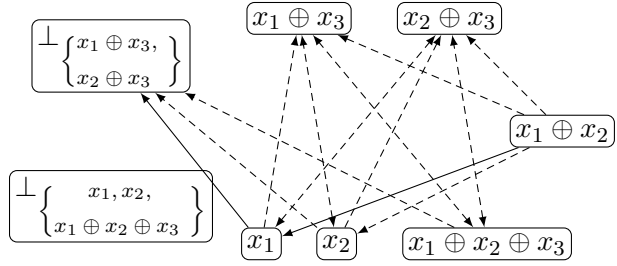
\begin{figure}
\centering
\begin{tikzpicture}[>=latex]
  \tikzstyle{vertex}=[draw, rounded corners=1mm, rectangle, fill=white,minimum size=12pt,inner sep=2pt]
  \tikzstyle{edge}=[densely dashed,->]
  \tikzstyle{doubleedge}=[densely dashed,<->]
  \tikzstyle{thickedge}=[->]
  \node[vertex] (G_1) at (-1,0.5) {$\bot_{\left\{\text{\parbox{4em}{$x_1\oplus x_3,$ \\ $x_2\oplus x_3$}}\right\}}$};
  \node[vertex] (G_2) at (-1,-1.5)   {$\bot_{\left\{\text{\parbox{6em}{$\quad\;\: x_1, x_2,$\\ $x_1\oplus x_2\oplus x_3$}}\right\}}$};
  \node[vertex] (G_3) at (1.5,1)  {$x_1\oplus x_3$};
  \node[vertex] (G_4) at (1,-2)  {$x_1$};
  \node[vertex] (G_5) at (3.5,1)  {$x_2\oplus x_3$};
  \node[vertex] (G_6) at (2,-2)  {$x_2$};
  \node[vertex] (G_7) at (4,-2)  {$x_1\oplus x_2\oplus x_3$};
  \node[vertex] (G_8) at (5,-0.5)  {$x_1\oplus x_2$};
  \draw[thickedge] (G_4) -> (G_1);
  \draw[edge] (G_4) -> (G_3);
  \draw[doubleedge] (G_4) -> (G_5);
  \draw[edge] (G_6) -> (G_1);
  \draw[doubleedge] (G_6) -> (G_3);
  \draw[edge] (G_6) -> (G_5);
  \draw[edge] (G_7) -> (G_1);
  \draw[doubleedge] (G_7) -> (G_3);
  \draw[doubleedge] (G_7) -> (G_5);
  \draw[edge] (G_8) -> (G_3);
  \draw[thickedge] (G_8) -> (G_4);
  \draw[edge] (G_8) -> (G_5);
  \draw[edge] (G_8) -> (G_6);
\end{tikzpicture}
\caption{The directed graph $G_s$ constructed when adding $x_1\oplus x_2$ to the minimal partition $\left\{\{x_1\oplus x_3, x_2\oplus x_3\}, \{x_1, x_2, x_1\oplus x_2\oplus x_3\}\right\}$. A minimum length path is shown in solid lines, resulting in the new partition $\left\{\{x_1\oplus x_3, x_2\oplus x_3, x_1\}, \{x_1\oplus x_2, x_2, x_1\oplus x_2\oplus x_3\}\right\}.$}
\label{fig:partition}
\end{figure}

Rather than generating the graph $G_s$ explicitly for each element $s$, we try to construct a path from $s$ to some $\bot_p$ breadth-first (Algorithm~\ref{alg:partition}). The time complexity of breadth-first search is $O(|E| + |V|)$ for a graph with edge set $E$ and vertices $V$. We can note that there are $|S'|+|P|+1$ vertices and at most $|S'|^2 + |P|\cdot(|S'| + 1) + (|S'|+|P|)$ edges in $G_s$, as well as the fact that $|P|\leq |S'|$. Since each edge requires a single test for independence in $O(n^3)$ time, the breadth first search requires time in $O(|S'|^2\cdot n^3 + |S'|).$

\begin{algorithm}
\caption{Matroid partitioning algorithm}
\label{alg:partition}
\centering
\begin{algorithmic}
	\Function{Partition}{$s, P, (S, I)$}
		\State{/* $I$ denotes the independence oracle,}
		\State{\hspace{0.7em} $P$ is a minimal partition */}
			\State Create path queue $Q$, $Q$.enqueue($s\to\emptyset$)
			\State Unmark each element of $S$, mark $s$
			\While{$Q$ non-empty}
				\State $t$ := $Q$.dequeue()
				\For{each $A\in P$}
					\State Set $A' := A\cup\{$head($t$)$\}$
					\If{$A'\in I$}
						\State Set $A := A'$
						\For{each $u\to v$ in path $t$}
							\State Replace $u$ with $v$ in its current partition
						\EndFor
					\Else
						\For{each unmarked $u\in A$}
							\If{$A'\setminus\{u\}\in I$}
								\State $Q$.enqueue($u\to t$)
								\State Mark $u$
							\EndIf
						\EndFor
					\EndIf
				\EndFor
			\EndWhile
				\State If no path was found, set $P := P\cup\{s\}$
	\EndFunction
\end{algorithmic}
\end{algorithm}

Algorithm~\ref{alg:partition} details the algorithm for adding an element $s$ to a partition $P$ of matroid $(S, I)$ -- the full algorithm follows by iteratively adding each element of the ground set to the initially empty partition, and correctness follows from the property that if $P$ is minimal for $(S, I)$, then the new partition $P'$ is minimal for $(S\cup\{s\}, I)$. Since adding a single element to a partition of $i$ elements takes $O\left((2i)^2\cdot n^3+2i\right)$ time, and $\sum_i^{|S|}i^2=\frac{|S|^3}{3}+\frac{|S|^2}{2}+\frac{|S|}{6}$, we see that Algorithm~\ref{alg:partition} can be used to partition the full set $S$ in $O(|S|^3\cdot n^3)$ time.

\section{Extending to a universal gate set}
\label{sec:uni}

In the previous sections we described a method for re-synthesizing $\{CNOT, T\}$ circuits that removes redundant $T$-gates by computing the total phase and parallelizing the phase gates through matroid partitioning. However, the usefulness of such an algorithm on its own is marred by the fact that $\{CNOT, T\}$ circuits are a restricted class of quantum circuits -- in particular, they do not create superpositions or interference between basis states. To apply the optimization procedure to more complex quantum circuits, we extend these ideas to deal with the universal gate set $\{H, CNOT, T\}.$

We recall that a Hadamard gate $H$ has the effect $$H:\ket{x_1}\mapsto\frac{1}{\sqrt{2}}\sum_{x_2\in \F}\omega^{4\cdot x_1 \cdot x_2}\ket{x_2}$$ on a basis state $x_1\in\F$. We call $x_2$ a {\it path variable}, following in the tradition of similar representations of quantum circuits \cite{dhmhno05, r09}, called {\it sum over path} representations. We note that the state $\ket{x_1}$ effectively ceases to exist, having been replaced with $\ket{x_2}$. Circuits over $\{H, CNOT, T\}$ can then be described by a phase polynomial and set of linear Boolean outputs, similar to Lemma~\ref{lem:Tpar}.

\begin{lemma}
\label{lem:sop}
If a unitary $U\in U(2^n)$ is exactly implementable by an $n$-qubit circuit over $\{H, CNOT, T\}$ with $k$ $H$ gates, then for $x_1x_2...x_n\in\F^n$,
$$U|x_1x_2...x_n\rangle = \frac{1}{\sqrt{2^k}}\hspace{-2em}\sum_{\qquad x_{n+1}...x_{n+k}\in\F^k}\hspace{-2em}\omega^{p(x_1,x_2, ..., x_{n+k})}\ket{y_1y_2...y_n}$$
where $y_i=h_i(x_1, x_2, \dots, x_{n+k})$ and 
\begin{align*}
p(x_1, &x_2, ..., x_{n+k})= \sum_{i=1}^lc_i\!\cdot\! f_i(x_1, ..., x_{n+k})+4\cdot\sum_{i=1}^k x_{n+i}\!\cdot\! g_i(x_1, ..., x_{n+k})
\end{align*}
for some linear Boolean functions $h_i, f_i, g_i$ and coefficients $c_i\in\mathbb{Z}_8$. The $k$ path variables \\ $x_{n+1},...,$ $x_{n+k}$ result from the application of Hadamard gates.
\end{lemma}
\begin{proof}
Follows from the effect of each gate on the computational basis states.
\end{proof}

It can be noted that unlike Lemma~\ref{lem:Tpar}, the converse is not true -- some computations of the form in Lemma~\ref{lem:sop} do not define unitary transformations.

Synthesis of a circuit based on the above representation is more challenging. In particular, the Hadamard gates in effect destroy values and create new ones, changing the state space to some new subspace of $\F^{n+k}$, possibly with greater dimension if the destroyed value was linearly dependent with the other qubits (e.g. an initialized ancilla). Each linear Boolean function is likewise computable only in some of the possible state spaces. To re-synthesize we then need to apply Hadamard gates in such a way as to be able to pick up each phase factor and end in the state space $\spn\left(\{y_1,y_2,...,y_n\}\right)$.

Since the application of a Hadamard gate changes the state space, the state of each qubit must be chosen so that the qubits span a suitable space afterwards. As an illustration consider the transformation $$\ket{x_1x_2}\mapsto\frac{1}{\sqrt{2}}\sum_{x_3\in\F}\omega^{4\cdot x_3\cdot x_2}\ket{(x_1\oplus x_2)x_3}.$$ We could achieve the correct phase by applying a Hadamard gate on the second qubit first, but the resulting state would be $\ket{x_1x_3}$, from which we cannot directly construct the output state $\ket{(x_1\oplus x_2)x_3}$. 
The simplest way to choose a suitable state for each qubit before applying a Hadamard gate is to use the qubit's state in the original circuit, though there may be other ways of computing such states. During re-synthesis we then transform the state to match the state in the original circuit before a particular Hadamard gate is applied.

In this sense, the Hadamard gates are fixed by the original circuit and the re-synthesis process needs to place the remaining terms of the phase (i.e. $c_i\cdot f_i(x_1, x_2, ..., x_{n+k})$ in between them. One approach is to use the greedy nature of Algorithm~\ref{alg:partition} to maintain a partition of those functions that are computable in the current state space of the circuit. Specifically, we iterate through the Hadamard gates and for each one we partition any elements that are in the new state space -- this step relies on the fact that Algorithm~\ref{alg:partition} is greedy to avoid having to repartition the elements that were already in the old state space. For any block in the partition containing functions that will not be computable after the next Hadamard gate, we remove the block and synthesize a $\{CNOT, T\}$ circuit applying those phases. A more detailed description is given in Section~\ref{sec:tpar}.

While this method is heuristic, we note that the partition is always minimal for the set of currently computable functions. In particular, Algorithm~\ref{alg:partition} produces a minimal partition when given a minimal partition, and removing blocks from a partition does not affect minimality -- given a subset $P'$ of a minimal partition $P$, if there existed a partition $P_0$ of the elements in $P'$ such that $|P_0|<|P'|$, then $P_0\cup P\setminus P'$ is a partition of the elements in $P$ into fewer sets.

One problem arises in that the dimension of the state space may increase in the next subcircuit (e.g. if the Hadamard is applied to an ancilla qubit). In this case, the independence condition~(\ref{eq:indcond}) of the matroid changes, and previous partitions may now be invalid under the new inequality. However, as a trivial consequence of the fact that the dimension increases by at most 1, a partition that is no longer independent can be modified to satisfy it by removing exactly one linearly dependent element. Furthermore, if all partitions are modified to satisfy the new independence condition in this way, the new partition is minimal and the elements that were removed can be repartitioned by Algorithm~\ref{alg:partition}.

\begin{lemma}
\label{lem:minimality}
Given a subspace $V$ of $\F^{n}$ with $\dim(V)=m$ and a set of linear Boolean functions $S\subseteq V^*$, let $$I_i=\{A\subseteq S | i - \rk\left(A\right) \leq n - |A|\}.$$

If $P$ is a minimal partition of $(S, I_m)$, then the partition $P'$ defined by removing one linearly dependent element from every $A\in P$ if $A\notin I_{m+1}$ is a minimal partition of $(S', I_{m+1})$, where $S'$ is the set of elements partitioned by $P'$.
\end{lemma}
\begin{proof} Suppose there exists some partition $P_0$ of $(S', I_{m+1})$ with $|P_0|<|P'|$. We then see that one element of $S\setminus S'$ can be added to any $A\in P_0$ to give a set in $I_m$. In particular, consider any $A\in P_0$. Since $m+1 - \rk(A) \leq n - |A|$ we see that $$m - \rk(A) \leq n - |A| - 1 = n - |A\cup\{s\}|$$ for any $s\in S\setminus S'$. We also note that $n - |A\cup\{s\}|\geq 0$ as required, since any subset of $S$ has rank at most $m$, so for any $A\in I_{m+1}$, $d+1 - \rk(A)>0$ and thus $|A|< n$. Therefore, for any $A\in P_0, s\in S\setminus S'$ we have $A\cup\{s\} \in I_m$.

Next we note that for any $T\subseteq S$ there exists a partition of $(T, I_m)$ with size at most $|T|-1$. This is a simple result of the fact that $m<n$, as any subset $A\subseteq T$ of size $2$ has rank at least $1$, so $$m - \rk(A)\leq m - 1 \leq n - 2 = n - |A|.$$ Additionally, any size $1$ subset of $T$ is trivially independent under $I_m$.

Thus, since we can add one element to every partition in $P_0$, and we can partition the remaining $|S\setminus S'| - |P_0|$ elements into at most $|S\setminus S'| - |P_0| - 1$ partitions, we see that there exists a partition of $(S, I_m)$ with size at most $$|P_0| + (|S\setminus S'| - |P_0| - 1) = |S\setminus S'| - 1.$$ Since $|S\setminus S'| - 1 < |P|$ we obtain a contradiction.
\end{proof}

\section{The $T$par algorithm}
\label{sec:tpar}

In the last section we described a heuristic for optimizing $T$-count and depth over the gate set $\{H, CNOT, T\}$. In this section, we present the concrete algorithm, $T$par, and enlarge the gate set to include the Pauli gates $$X:=\begin{pmatrix} 0 & 1 \\ 1 & 0 \end{pmatrix}, Y:=\begin{pmatrix} 0 & i \\ -i & 0 \end{pmatrix}, Z:=\begin{pmatrix} 1 & 0 \\ 0 & -1 \end{pmatrix}.$$ We ignore the irrelevant global phase $i$ in $Y=iXY$, though it can be recovered by applying $XPXP=iI$ to any qubit. Appendix~\ref{app:example} gives a demonstration of the $T$par algorithm on a simple circuit.

Recall that in the computational basis, the input space of a circuit with $n$ qubits, $n-m$ ancillae and $k$ Hadamard gates is a dimension $m$ subspace $V$ of $\F^{n+k}$. We represent the state of a qubit as a vector in the dual space of $\F^{n+k}$, $\F^{(n+k)*}$, along with a Boolean value $b$, called the parity, which is used to record bit flips. Specifically, we note that $X:\ket{x}\mapsto\ket{1\oplus x}$, so we can model bit flips with a single parity bit. We denote the set of states $\F\times\F^{(n+k)*}$ as $\mathcal{S}$.

Given a Clifford + $T$ circuit $C$, written as a sequence of gates over $\{X, Y, Z, P, P^\dagger, H, CNOT, T,$ $T^\dagger\}$, we first compute a triple $\langle S, Q, H\rangle\in\mathcal{D}$ representing $C$. $S=\{(c, f) | c\in\mathbb{Z}_8, f\in\mathcal{S}\}$ stores the $T^k$ phase factors as linear Boolean functions with parity and multiplicity, $Q=(g_1,g_2,...,g_n)\in\mathcal{S}^n$ tracks the state of each qubit, and $H=(h_1,h_2,...h_k)$ gives a sequence of Hadamard gates where each entry $h_i$ stores the input and output states, $h_i.Q_I$ and $h_i.Q_O$ respectively. We define the initial state of the circuit as $Q_0=\left((0, x_1),(0, x_2),...,(0, x_m), (0, 0), ..., (0, 0)\right)$, which is understood as the state $\ket{x_1x_2...x_m}\ket{0}^{\otimes n-m}$. To compute $\langle S, Q, H\rangle$, we use the function $\llbracket U \rrbracket:\mathcal{D}\to\mathcal{D}$ (Figure~\ref{fig:gates}) to sequentially update the triple $\langle S, Q, H\rangle$ for each gate $U$ in the circuit, starting from the initial value $\langle\emptyset, Q_0, \emptyset\rangle$.

\begin{figure}[h]
\begin{align*}
	&\llbracket X_i \rrbracket \langle S, Q, H\rangle = \langle S, (g_1,...,g_{i-1},1\oplus g_i,...,g_n), H\rangle \\
	&\llbracket Z_i \rrbracket \langle S, Q, H\rangle = \langle S\uplus\{(4, g_i)\}, Q, H\rangle \\
	&\llbracket Y_i \rrbracket \langle S, Q, H\rangle = \langle S\uplus\{(4, g_i)\}, Q', H\rangle 
		\hspace{3.2em}\textrm{where  $Q'=(g_1,...,g_{i-1},1\oplus g_i,...,g_n)$} \\
	&\llbracket P_i \rrbracket \langle S, Q, H\rangle) = \langle S\uplus\{(2, g_i)\}, Q, H\rangle \\
	&\llbracket P_i^\dagger \rrbracket \langle S, Q, H\rangle) = \langle S\uplus\{(6, g_i)\}, Q, H\rangle \\
	&\llbracket H_i \rrbracket \langle S, Q, (h_1,h_2,...,h_j)\rangle) = \langle S, Q', H'\rangle\;\; 
		\hspace{1.5em}\textrm{where $Q'=(g_1,...,g_{i-1},(0, x_{j+i}),...,g_n)$} \\
		&\hspace{22em}\textrm{$H'=(h_1,h_2,...,h_j,\{Q_I=Q, Q_O=Q'\})$} \\
	&\llbracket CNOT_{(i, j)} \rrbracket \langle S, Q, H\rangle) = \langle S, Q', H\rangle 
		\hspace{3.9em}\textrm{where $Q'=(g_1,...,g_{j-1}, g_j\oplus g_i,...,g_n)$} \\
	&\llbracket T_i \rrbracket \langle S, Q, H\rangle) = \langle S\uplus\{(1, g_i)\}, Q, H\rangle \\
	&\llbracket T_i^\dagger \rrbracket \langle S, Q, H\rangle) = \langle S\uplus\{(7, g_i)\}, Q, H\rangle \\
\end{align*}
\vspace{-2em}
\caption{Semantic function $\llbracket\cdot\rrbracket$. We define $S\uplus T$ as the union of $S$ and $T$ where any $f$ such that $(c_1, f)\in S, (c_2, f)\in T$ is given coefficient $c_1+c_2\mod 8$. $U_i$ denotes the gate $U$ applied to qubit $i$ and $CNOT_{(i, j)}$ specifies $i$ as the control qubit and $j$ as the target.}
\label{fig:gates}
\end{figure}

The $T$par algorithm (Algorithm~\ref{alg:Tpar}) proceeds as follows: after computing $\langle S, Q, H\rangle$, we synthesize a new circuit by iterating through the Hadamard gates in $H$ while updating a partition $P$ of the functions of $S$ that are computable in the current subcircuit. In particular, we divide $S$ into $S_P$ and $S_{-P},$ where $S_P$ are the already partitioned elements and $S_{-P}$ are those not already partitioned. For a given $h_i=\{Q_I, O_O\}$, for every $(c, f)\in S_{-P}$ we check whether $f\in \spn(Q_I)$. If so, we add $(c, f)$ to the current partition using Algorithm~\ref{alg:partition} with the independence relation $A\subseteq S\in I$ if and only if  $\rk (Q_I) - \rk(A)\leq n - |A|.$ After partitioning, we update $S_P$ and $S_{-P}$ accordingly. The tests for inclusion of each function $f$ in $\spn(h_i.Q_I)$ requires $O\left(|S_{-P}|\cdot(n+k)^3\right)$ time, and we can loosely bound the partitioning time as the time to partition the entire set $S$ using Algorithm~\ref{alg:partition}, $O(\left(|S|^3\cdot(n+k)^3\right)$; since $|S_{-P}|\leq |S|$ the entire step thus takes $O\left(|S|^3\cdot(n+k)^3\right)$ time. A tighter analysis is possible, though we omit it as the algorithm is heuristic in nature.

We next iterate through $P$ and for each block $A\in P$, if $f\notin \spn(Q_O)$ for some $(c, f)\in A$ we remove $A$ from the partition and synthesize a circuit computing the relevant phase factors. While we defer the discussion of the synthesis procedure for now, we remark that it requires $O\left((n+k)^3\right)$ time. Otherwise, if $A$ is no longer an independent set under the new independence relation $A\subseteq S\in I$ if and only if $ \rk (Q_O) - \rk(A)\leq n - |A|,$ we remove a linearly dependent element from $A$ and update $S_P$ and $S_{-P}$ so that the deleted element will be re-partitioned on the next iteration. As $\rk(A)$ and a linearly dependent element can both be found with one application of Gaussian elimination, this step also requires $O\left((n+k)^3\right)$ time, and so the entire loop takes $O\left(|P|\cdot(n+k)^3\right)$ time.

Finally, we synthesize a circuit applying the Hadamard gate in $O\left((n+k)^3\right)$ time, and repeat the entire process for the next Hadamard. The entire algorithm, shown in Algorithm~\ref{alg:Tpar}, thus runs in time $$O\left(|C|\cdot n+k\cdot(n+k)^3\cdot\left(|S|^3+|P| + 1\right)\right).$$ As $|C|\cdot n$ is in most cases negligible compared to the $k\cdot(n+k)^3\cdot|S|^3$ factor, and $|P|\leq|S|$, we describe the runtime as simply $O\left(k\cdot|S|^3\cdot (n+k)^3\right)$.  Moreover, it should be noted that if no repartitioning is done, the runtime is bounded by $O\left(|S|^3\cdot (n+k)^3\right)$, as each element is partitioned exactly once, rather than the worst case of $k$ times in general.

\begin{algorithm}
\caption{T-parallelization algorithm}
\label{alg:Tpar}
\centering
\begin{algorithmic}
	\Function{Tpar}{Clifford + $T$ circuit $C$}
		\State $C':=\emptyset$
		\State $\langle S, Q, H\rangle:=\llbracket{C_{|C|}}\rrbracket\cdots\llbracket{C_1}\rrbracket\langle\emptyset, Q_0, \emptyset\rangle$
		\State Set $S_P:=\emptyset;\; S_{-P}:=S;\; P:=\emptyset$
		\For{each $1\leq i\leq k$}
			\State $I:=\{A\subseteq S | \rk(h_i.Q_I) -\rk(A)\leq n - |A|\}$
			\For{each $(c, f)\in S_{-P}$}
				\If{$f\in\spn(h_i.Q_I)$}
					\State $P:=$Partition$\left((c, f), P, (S_P, I)\right)$
					\State $S_P:= S_P \cup \{(c, f)\};\; S_{-P}:= S_{-P} \setminus \{(c, f)\}$ 
				\EndIf
			\EndFor
			\For{each $A\in P$}
				\If{$i = k$ or $\exists (c, f)\in A$ s.t. $f\notin\spn(h_i.Q_O)$ \\ \qquad\qquad}
					\State Append($C'$, Synthesize($A, h_i.Q_I, h_i.Q_I$))
					\State $P:=P\setminus A$
				\ElsIf{$\rk(h_i.Q_O) - \rk(A) > n - |A|$}
					\State Choose $(c, f)\in A$ such that $\rk(A) = \rk(A\setminus\{(c, f)\})$
					\State $A:=A\setminus \{(c, f)\}$
					\State $S_P:=S_P\setminus\{(c, f)\}; \;S_{-P}:=S_{-P}\cup\{(c, f)\}$
				\EndIf
			\EndFor
			\State Append($C'$, Synthesize($\emptyset, h_i.Q_I, h_i.Q_O$))
		\EndFor
		\State \Return $C'$
	\EndFunction
	\vspace{1em}

	\Function{Synthesize}{$A, Q_I, Q_O$}
			\State /* Synthesize a circuit implementing $U:\ket{Q_I}\mapsto\omega^{\sum_{(c, (b, f))\in A}c\cdot b\oplus f(x_1,x_2,...,x_{n+k})}\ket{Q_O}$ */
			\State Compute $A'\supseteq A$ s.t. $\rk(A') = \rk(Q_I), |A'|=n$
			\State Synthesize $\{CNOT, X\}$ circuit $C_1:\ket{Q_I}\mapsto \ket{A'}$
			\State Synthesize $\{Z, P, T\}$ circuit $C_2:\ket{A'}\mapsto \omega^{\sum_{(c,b,f)\in A'}c\cdot b\oplus f(x_1,x_2,...,x_{n+k})}\ket{A'}$
			\State Synthesize $\{CNOT, X, H\}$ circuit $C_3:\ket{A'}\mapsto \ket{Q_O}$
			\State Return $C_1C_2C_3$
	\EndFunction
\end{algorithmic}
\end{algorithm}

\subsection{Synthesizing partitions}

We now describe the general synthesis procedure, S{\footnotesize YNTHESIZE}($A, Q_I, Q_O$), from Algorithm~\ref{alg:Tpar}. The procedure synthesizes a circuit with inputs $Q_I$ and outputs $Q_O$ applying the phases given by a computable partition $A$ of linear Boolean functions.

The algorithm proceeds by first extending $A$ with $n - |A|$ linear Boolean functions to form a set $A'$ with rank equal to $\rk(Q_I)$ -- this is accomplished by using row operations to reduce $A$ to a subset of $Q_I$, then adding the vectors in $Q_I\setminus A$. Next we synthesize a circuit computing $A'$ by reducing $Q_I$ and $A'$ to the same basis using Gaussian elimination in $O((n+k)^3)$ time, where addition of two rows corresponds to the application of a $CNOT$ gate, and parity changes correspond to $X$ gates. The circuit reducing $Q_I$ to this basis is applied forwards, while the circuit reducing $A'$ is applied in reverse, giving a circuit mapping $\ket{Q_I}\mapsto\ket{A'}$. The phase factors are applied by constructing a combination of $T, P$ and $Z$ gates, corresponding to the relevant coefficients, then the circuit mapping $\ket{Q_I}$ to $\ket{A'}$ is reversed to compute $\ket{Q_I}$. In the case when $Q_O\neq Q_I$, the corresponding Hadamard gate is applied to compute the output $\ket{Q_O}$.

As alluded to before, we now see that the synthesis procedure has time complexity $O\left((n+k)^3\right)$ since it requires a constant number (three) of applications of Gaussian elimination.  Moreover, the $T$-depth of the resulting circuit is 1.

As an important practical issue, Gaussian elimination based synthesis produces linear reversible circuits that are non-optimal in terms of the number of gates or depth, resulting in a potential increase in the number of $CNOT$ gates after re-synthesizing, as compared to the original design.  While our focus was on the optimization of $T$ gates, there exist algorithms, \cite{m07, pmh08}, that produce more efficient circuits for linear reversible functions.  Specifically, \cite{pmh08} provides an algorithm to synthesize linear reversible circuits with $\Theta(n^2/\log(n))$ gates, and \cite{m07} reports an O($n$)-depth algorithm. More recently, \cite{km13} described an optimization procedure for stabilizer circuits that could be applied afterward to further optimize linear reversible and Clifford group subcircuits.  In practical implementations an advanced linear reversible synthesis algorithm should be used.  Compared to the optimization of $T$-depth, we considered the optimization of the linear reversible circuit stages to be a second order improvement and did not pursue it in this work.

\section{Results}
\label{sec:res}
\vspace{1mm}

We implemented Algorithm~\ref{alg:Tpar} in C++\footnote{C++ Source code available at \href{http://code.google.com/p/tpar/}{http://code.google.com/p/tpar/}.} and used it to optimize various quantum circuits, specifically arithmetic and reversible ones, from the literature.  Individual circuits were written in the standard fault-tolerant universal gate set $\{X, Y, Z, H, P, P^\dagger, CNOT, T, T^\dagger\}$, using the decompositions found in \cite{ammr13} where applicable.  As most arithmetic circuits are dominated by Toffoli gates, we used the lowest $T$-depth implementation of the Toffoli without ancillae known \cite{ammr13}.

Results are reported in Tables~\ref{tab:benchmarks1} and \ref{tab:benchmarks2}.  They were generated in Debian Linux running on a quad-core 64-bit Intel Core i7 2.40GHz processor and 8 GB RAM. Table~\ref{tab:benchmarks1} gives gate counts for the circuits before and after optimization.  Table~\ref{tab:benchmarks2} shows $T$-depths before and after optimization using either 0, $n$, or $\infty$ ancillae\footnote{In order to give an example of the trade-off between number of ancillae and the $T$-depth for some non-constant number of ancillae, we arbitrarily chose to illustrate results with $n$ ancillae. Our implementation allows any other computable value.} where $n$ denotes the original number of qubits in the circuit.  The $T$-depth for each circuit before optimization was computed by parallelizing the $T$-gates and Toffoli gates by hand, and writing each group of parallel Toffoli gates in $T$-depth 3.

With no extra ancillae added, all the tested benchmarks show significant reductions in terms of both $T$-count and $T$-depth, with average reductions of $39.9\%$ and $54.3\%$ respectively. The algorithm is particularly effective in cases where adjacent Toffoli gates share either controls or targets, as many of the phases cancel -- each of the GF($2^m$) multipliers share this structure, and as a result show large reductions in $T$-count and $T$-depth. In fact, the $T$par algorithm will parallelize any GF($2^m$) multiplication circuit to $T$-depth 2 when given sufficiently many ancillae, by noting that each such circuit contains two stages of Toffoli gates that result in one $\{CNOT, T\}$ stage each after removing trivial identities. By comparison, since the Toffoli gates {\it cannot} be all written in parallel, the minimum $T$-depth achievable using $T$-depth 1 Toffoli implementations \cite{s13} is $12(m-1)$. Those circuits that mix controls and targets between adjacent Toffoli gates are less affected by the optimization, e.g., CSUM-MUX$_9$, as the Hadamard gates create barriers to $T$ parallelization.

The runtimes show that algorithm scales well to large circuits, the largest tested circuit having 192 qubits, 28672 $T$ gates and 8192 Hadamard gates. This stands in contrast to most previous efforts to optimize quantum circuits, which have generally been limited in usefulness to a few qubits and a small number of gates. While the inclusion of Hadamard gates adds significant complexity to the algorithm, it has actually reduced runtime of the algorithm compared to experiments where $T$par was applied only to $\{CNOT, T\}$ subcircuits which is likely a result of the greater $T$-count reductions. As a result, $T$par appears to be an effective heuristic algorithm for the large-scale optimization of fault-tolerant circuits.

We also tested our algorithm's ability to make use of ancillae to optimize $T$-depth (Table~\ref{tab:benchmarks2}). For each of the benchmark circuits, we applied our algorithm with an added $n$ ancillae, where the original circuit contained $n$ qubits. We also report the minimum $T$-depth achievable for each circuit using our algorithm. It can be noted that our algorithm usually decreases in running time when ancillae are added, due to the reduced number of partitions and thus faster matroid partitioning. Specifically, when there are many ancillae, for the majority of the time when an item $s$ is partitioned it can be directly added to one of the partitions in time $O(|P|\cdot(n+k)^3)$. The algorithm is thus very flexible, and the experimental data illustrates a great potential for exploring space-time trade-off in quantum circuits. 

\begin{table}[t]
\scriptsize
\centering
\begin{tabular}{|l|c|c|c|c||r|c|c|c|c|}
\hline
Benchmark&$N$&$x_C$&$x_T$&$x_g$&$x'_C$&$x'_T$&$x'_g$&Time (s)&$T$-count reduction (\%) \\ \hline\hline
Mod 5$_{4}$ \cite{m11}&5&32&28&9&48&16&12&0.000&42.9\\ \hline
VBE-Adder$_{3}$ \cite{vbe96}&10&80&70&20&114&24&23&0.001&65.7\\ \hline
CSLA-MUX$_{3}$ \cite{vi05}&15&90&70&20&425&62&21&0.001&11.4\\ \hline
CSUM-MUX$_{9}$ \cite{vi05}&30&196&196&84&411&112&70&0.005&42.9\\ \hline
QCLA-Com$_{7}$ \cite{dkrs06}&24&215&203&73&583&95&73&0.003&53.2\\ \hline
QCLA-Mod$_{7}$ \cite{dkrs06}&26&441&413&132&1185&249&138&0.008&39.7\\ \hline
QCLA-Adder$_{10}$ \cite{dkrs06}&36&273&238&86&737&162&73&0.018&31.9\\ \hline
Adder$_{8}$ \cite{ttk10}&24&466&399&126&920&215&153&0.007&46.1\\ \hline
RC-Adder$_{6}$ \cite{cdkp04}&14&104&77&30&234&63&29&0.001&18.2\\ \hline
Mod-Red$_{21}$ \cite{ms12}&11&121&119&58&301&73&51&0.001&38.7\\ \hline
Mod-Mult$_{55}$ \cite{ms12}&9&55&49&36&166&37&20&0.000&24.5\\ \hline
$\Lambda_3(X)$ -- \cite{bbcdmsssw95}&5&28&28&8&54&16&12&0.000&42.9\\ \hline
\hspace{3.05em} -- \cite{nc00}&5&21&21&6&41&15&9&0.000&28.6\\ \hline
$\Lambda_4(X)$ -- \cite{bbcdmsssw95}&7&56&56&16&90&28&23&0.000&50.0\\ \hline
\hspace{3.05em} -- \cite{nc00}&7&35&35&10&63&23&16&0.000&34.3\\ \hline
$\Lambda_5(X)$ -- \cite{bbcdmsssw95}&9&84&84&24&132&40&34&0.001&52.4\\ \hline
\hspace{3.05em} -- \cite{nc00}&9&49&49&14&94&31&23&0.000&36.7\\ \hline
$\Lambda_{10}(X)$ -- \cite{bbcdmsssw95}&19&224&224&64&328&100&89&0.004&55.4\\ \hline
\hspace{3.5em} -- \cite{nc00}&19&119&119&34&232&71&58&0.002&40.3\\ \hline
GF($2^4$)-Mult \cite{mmcp09}&12&115&112&32&324&68&27&0.001&39.3\\ \hline
GF($2^5$)-Mult \cite{mmcp09}&15&179&175&50&535&111&36&0.004&36.6\\ \hline
GF($2^6$)-Mult \cite{mmcp09}&18&257&252&72&649&150&43&0.008&40.5\\ \hline
GF($2^7$)-Mult \cite{mmcp09}&21&349&343&98&992&217&36&0.031&36.7\\ \hline
GF($2^8$)-Mult \cite{mmcp09}&24&468&448&128&1256&264&40&0.052&41.1\\ \hline
GF($2^9$)-Mult \cite{mmcp09} &27&575&567&162&1701&351&44&0.110&38.1\\ \hline
GF($2^{10}$)-Mult \cite{mmcp09}&30&709&700&200&2176&410&69&0.227&41.4\\ \hline
GF($2^{16}$)-Mult \cite{mmcp09}&48&1856&1792&512&6592&1040&82&5.079&42.0\\ \hline
GF($2^{32}$)-Mult \cite{mmcp09}&96&7291&7168&2048&33269&4128&166&602.577&42.4\\ \hline
GF($2^{64}$)-Mult \cite{mmcp09}&192&28860&28672&8192&180892&16448&334&95447.466&42.6\\ \hline\hline
\multicolumn{9}{ |l| }{Average}&39.9\\ \hline
\multicolumn{9}{ |l| }{Maximum}&65.7\\ \hline
\end{tabular}
\caption{$T$-count benchmarks. We report the gate counts after optimizing circuits with $T$par, using no extra ancillae. $N$ specifies the number of qubits. $x_C$ reports the number of $CNOT$ gates, $x_T$ reports the number of $T$ gates and $x_U$ reports the number of other gates. $x'$ denotes the number of gates after optimization.}
\label{tab:benchmarks1}
\end{table}

\begin{table}[t]
\scriptsize
\centering
\begin{tabular}{|l|c||c|c||r|c|c||r|c|c|}
\hline
Benchmark&$T$-depth&$T$-depth&Red.&$T$-depth&Time&Red.&$T$-depth&Time&Red. \\
&original&0 ancilla&(\%)&$N$ ancilla&(s)&(\%)&$\infty$ ancilla&(s)&(\%) \\ \hline\hline
Mod 5$_{4}$ \cite{m11}&12&6&50.0&3&0.000&75.0&3&0.000&75.0\\ \hline
VBE-Adder$_{3}$ \cite{vbe96}&24&9&62.5&5&0.000&79.2&5&0.000&79.2\\ \hline
CSLA-MUX$_{3}$ \cite{vi05}&21&8&61.9&4&0.004&81.0&4&0.001&81.0\\ \hline
CSUM-MUX$_{9}$ \cite{vi05}&18&9&50.0&4&0.003&77.8&3&0.005&83.3\\ \hline
QCLA-Com$_{7}$ \cite{dkrs06}&27&12&55.6&7&0.003&74.1&7&0.004&74.1\\ \hline
QCLA-Mod$_{7}$ \cite{dkrs06}&57&29&49.1&14&0.008&75.4&14&0.010&75.4\\ \hline
QCLA-Adder$_{10}$ \cite{dkrs06}&24&11&54.2&6&0.005&75.0&6&0.006&75.0\\ \hline
Adder$_{8}$ \cite{ttk10}&69&30&56.5&15&0.007&78.3&15&0.008&78.3\\ \hline
RC-Adder$_{6}$ \cite{cdkp04}&33&22&33.3&11&0.002&66.7&11&0.001&66.7\\ \hline
Mod-Red$_{21}$ \cite{ms12}&48&25&47.9&15&0.002&68.8&15&0.011&68.8\\ \hline
Mod-Mult$_{55}$ \cite{ms12}&15&7&53.3&4&0.000&73.3&4&0.005&73.3\\ \hline
$\Lambda_3(X)$ -- \cite{bbcdmsssw95}&12&8&33.3&4&0.001&66.7&4&0.000&66.7\\ \hline
\hspace{3.05em} -- \cite{nc00}&9&6&33.3&3&0.000&66.7&3&0.001&66.7\\ \hline
$\Lambda_4(X)$ -- \cite{bbcdmsssw95}&24&13&45.8&8&0.001&66.7&8&0.002&66.7\\ \hline
\hspace{3.05em} -- \cite{nc00}&15&9&40.0&5&0.001&66.7&5&0.000&66.7\\ \hline
$\Lambda_5(X)$ -- \cite{bbcdmsssw95}&36&18&50.0&12&0.001&66.7&12&0.004&66.7\\ \hline
\hspace{3.05em} -- \cite{nc00}&21&12&42.9&7&0.001&66.7&7&0.001&66.7\\ \hline
$\Lambda_{10}(X)$ -- \cite{bbcdmsssw95}&96&43&55.2&32&0.005&66.7&32&0.032&66.7\\ \hline
\hspace{3.5em} -- \cite{nc00}&51&27&47.1&17&0.003&66.7&17&0.008&66.7\\ \hline
GF($2^4$)-Mult \cite{mmcp09}&36&6&83.3&4&0.001&88.9&2&0.001&94.4\\ \hline
GF($2^5$)-Mult \cite{mmcp09}&48&9&81.3&5&0.002&89.6&2&0.002&95.8\\ \hline
GF($2^6$)-Mult \cite{mmcp09}&60&9&85.0&5&0.005&91.7&2&0.003&96.7\\ \hline
GF($2^7$)-Mult \cite{mmcp09}&72&12&83.3&7&0.026&90.3&2&0.004&97.2\\ \hline
GF($2^8$)-Mult \cite{mmcp09}&84&13&84.5&7&0.035&91.7&2&0.006&97.6\\ \hline
GF($2^9$)-Mult \cite{mmcp09} &96&15&84.4&7&0.058&92.7&2&0.010&97.9\\ \hline
GF($2^{10}$)-Mult \cite{mmcp09}&108&16&85.2&7&0.157&93.5&2&0.012&98.1\\ \hline
GF($2^{16}$)-Mult \cite{mmcp09}&180&24&86.7&12&3.128&93.3&2&0.061&98.9\\ \hline
GF($2^{32}$)-Mult \cite{mmcp09}&372&47&87.4&23&644.189&93.8&2&1.246&99.5\\ \hline
GF($2^{64}$)-Mult \cite{mmcp09}&756&94&87.6&44&127287.329&94.2&2&78.641&99.7\\ \hline\hline
\multicolumn{3}{ |l| }{Average}&61.1&\multicolumn{2}{ l| }{}&78.5&\multicolumn{2}{ l| }{}&80.7\\ \hline
\multicolumn{3}{ |l| }{Maximum}&87.6&\multicolumn{2}{ l| }{}&94.2&\multicolumn{2}{ l| }{}&99.7\\ \hline
\end{tabular}
\caption{$T$-depth benchmarks. We report the $T$-depth after no optimization (original) and after optimization with 0 (cf. Table~\ref{tab:benchmarks1}), $n$, or unbounded ancillae.}
\label{tab:benchmarks2}
\end{table}

As an important application, our experiments include instances of the multiple control Toffoli gates, $\Lambda_k(X)$, which are widely used in the construction of reversible circuits.  We report the results using two different implementations -- the Barenco {\em et al.} implementation using $k-2$ ancillae with arbitrary initial states  \cite{bbcdmsssw95}, and the Nielsen-Chuang implementation using $k-2$ ancillae initialized in the state $\ket{0}$ \cite{nc00}. In both constructions the ancillae are returned to their initial state. Our optimization of the Barenco {\em et al.} version reduces the $T$-count from $7(4k-8)$ to $3(4k-8) + 4$ and $T$-depth from $3(4k-8)$ to $4k-8$ with unbounded ancillae, in the instances we tried.  Likewise, our optimization of the Nielsen-Chuang implementation reduced the $T$-count from $7(2k-3)$ to $4(2k-3) + 3$ and $T$-depth from $3(2k-3)$ to $2k-3$. These formulae in fact hold for every $k\geq 3$, a result of the simple structure of the circuits. As the two versions use $4k-8$ and $2k-3$ sequential Toffoli gates, respectively, we note that $T$par parallelizes each Toffoli to $T$-depth 1 when sufficiently many ancillae are available.  Moreover, the reductions in $T$-count can be observed to correspond directly to each shared target or control -- in this way, $T$par achieves the same $T$ count and depth reductions as the multiply-controlled gate construction reported in \cite{s13}, but applies to more general cases and does not require implementing controls with the less intuitive $\Lambda_2(\pm iX)$ gates.

\begin{figure}[h]
\centering
\includegraphics[scale=0.2]{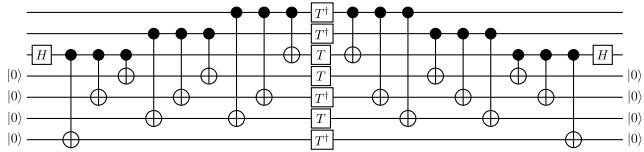}
\caption{$T$-depth 1 implementation of the Toffoli gate.}
\label{fig:tdepth1tof}
\end{figure}

\begin{figure}[h]
\centering
\includegraphics[scale=0.2]{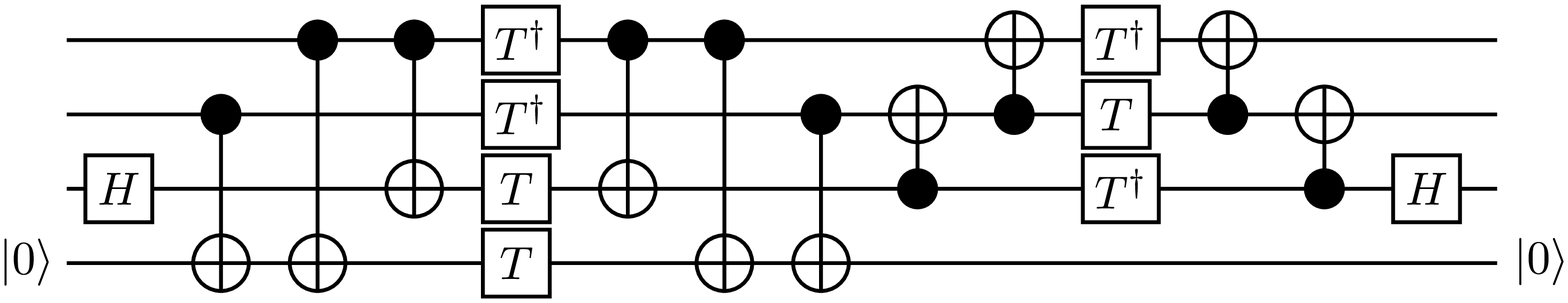}
\caption{$T$-depth 2 implementation of the Toffoli gate.}
\label{fig:tdepth2tof}
\end{figure}

\begin{figure}[h]
\centering
\includegraphics[scale=0.2]{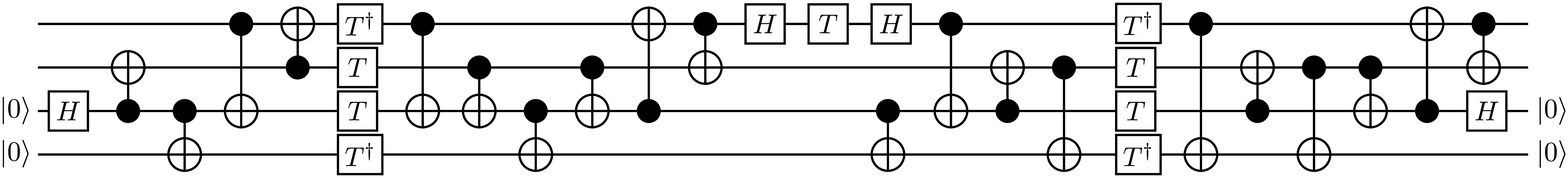}
\caption{$T$-depth 3 implementation of the controlled-$T$ gate (CNOT stages optimized by templates \cite{mdmn08}).}
\label{fig:tdepth3controlledt}
\end{figure}

\begin{figure}[h]
\centering
\includegraphics[scale=0.2]{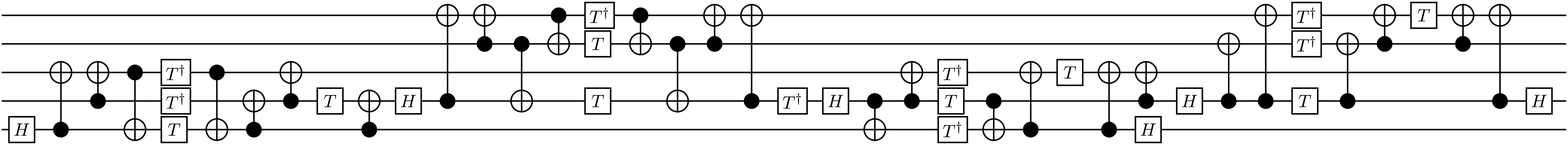}
\caption{An optimized implementation of the $\Lambda_3(X)$ gate (CNOT stages optimized by templates \cite{mdmn08}).}
\label{fig:tof4-tpar}
\end{figure}

As a final remark, we note that our algorithm reproduces many of the previous results regarding the optimization of $T$-depth. In particular, Figure~\ref{fig:tdepth1tof} shows the circuit produced by running $T$par on an implementation of the Toffoli gate. The circuit mirrors the $T$-depth 1 Toffoli reported in \cite{s13}.  Moreover, the full range of $T$-depths possible with different numbers of ancillae can be observed, as seen in Figure~\ref{fig:tdepth2tof}.  We also show a re-synthesized controlled-$T$ gate \cite{ammr13} using one ancilla to reduce the $T$-depth from 5 to 3 (Figure~\ref{fig:tdepth3controlledt}), and a re-synthesized Barenco {\em et al.} implementation of the $\Lambda_3(X)$ gate using no ancillae (Figure~\ref{fig:tof4-tpar}).

\section{Conclusion}
\label{sec:con}

We have described an algorithm for re-synthesizing Clifford + $T$ circuits with reduced $T$-count and depth. The algorithm uses a circuit representation based on linear Boolean functions, allowing $T$ gates to be combined and then parallelized through the use of matroid partitioning algorithms. The algorithm has worst case runtime that is cubic in the number of $T$ gates, qubits, and Hadamard gates, though our experiments show that the algorithm is sufficiently fast for practical circuit sizes.

Our benchmarks (Tables \ref{tab:benchmarks1} and \ref{tab:benchmarks2}) show that large gains can be obtained in reducing the $T$-count and $T$-depth of quantum circuits. In some cases, $T$-count was reduced by as much as $65.7\%$, while the $T$-depth could be reduced by up to $87.6\%$ without ancillae. Furthermore, the benchmarks illustrate that ancillae can be used to parallelize $T$ gates further, and given the runtimes reported the algorithm can be seen to provide substantial flexibility in exploring the trade-off between ancilla usage and $T$-depth. In the most extreme case we were able to reduce the $T$-depth of GF($2^{m}$)-Mult circuits from $12(m-1)$ to a constant of $2$,  using unbounded ancillae.  While the benchmarks were all arithmetic or otherwise reversible operations, such operations typically require the majority of the resources in circuits for quantum algorithms of interest \cite{IARPAQCS}.

We close by noting that as a consequence of the $T$par algorithm, reducing the number of terms in the mixed arithmetic polynomials describing the phase corresponds directly to reducing the $T$-complexity of quantum circuits; in fact, it was observed that minimization of $T$-count in $\{CNOT, T\}$ circuits is equivalent to minimizing the number of odd coefficients in the phase polynomial. A natural avenue of future work is then to develop methods for optimizing such polynomials for $T$-count and depth. This work also represents the first instance, to the authors' knowledge, of the use of sum over paths style representations in quantum circuit synthesis and optimization. While this representation has proven effective in optimizing circuit $T$-count and $T$-depth, the questions of synthesis for more general phases, e.g. the sum over paths representation of $\{H, CNOT, T\}$, and of the precise form of phases synthesizeable over the ``Clifford + $T$" gate set remain. Moreover, we leave it as a topic of future research to find new applications to optimization over different gate sets and cost metrics. Efficient practical synthesis of linear reversible circuits is another important direction that would directly contribute to improving the results of this work. 



\appendix

\section{Parallelizing $(\Lambda_2(X)\otimes I)(I\otimes \Lambda_2(X))$}
\label{app:example}

In this section we illustrate the workings of the $T$par algorithm using the following circuit:

\centerline{
\includegraphics[scale=0.3]{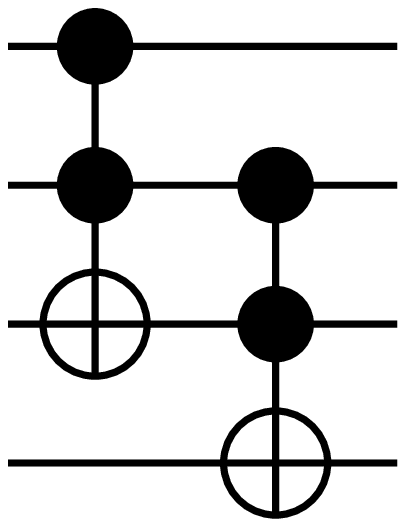}
}
\vspace{1em}

We first expand this circuit by using the $T$-depth 3 Toffoli gate implementation \cite{ammr13}:

\vspace{.5em}
\centerline{\includegraphics[scale=0.2]{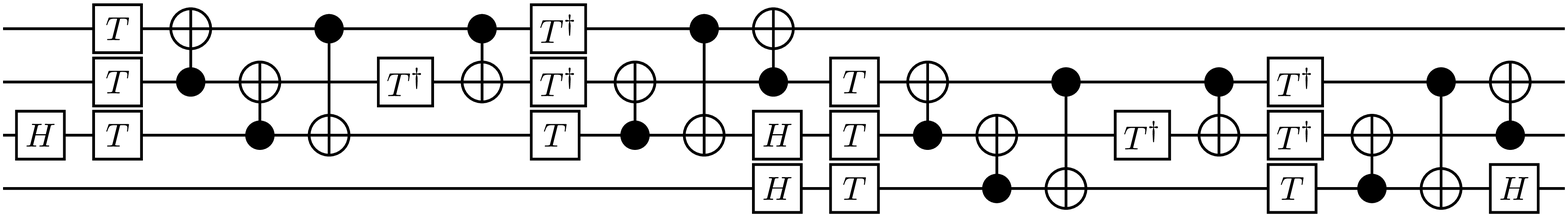}}
\vspace{.5em}

Next we compute $\langle S, Q, H\rangle$ by applying the function $\llbracket\cdot\rrbracket$ to each gate in sequence, starting with $\langle \emptyset, \left(x_1, x_2, x_3, x_4)\right), \emptyset\rangle$. The result is $$S = \left\{\text{\parbox{0.65\textwidth}{\centering $x_1, 2\cdot x_2, x_5, 7\!\cdot\!(x_1\oplus x_2), 7\!\cdot\!(x_1\oplus x_5), 7\!\cdot\!(x_2 \oplus x_5), (x_1\oplus x_2\oplus x_5),$ \\ $ x_6, x_7, 7\!\cdot\!(x_2\oplus x_6),7\!\cdot\!(x_2\oplus x_7), 7\!\cdot\!(x_6\oplus x_7),  (x_2\oplus x_6\oplus x_7)$}}\right\},$$ $$Q=(x_1,x_2,x_6,x_8),$$ $$H=(h_1,h_2,h_3,h_4) \text{\quad where}$$

\vspace{-2em} \begin{align*} 
h_1&=\left\{Q_I=(x_1,x_2,x_3,x_4), Q_O=(x_1,x_2,x_5,x_4)\right\}, \\
h_2&=\left\{Q_I=(x_1,x_2,x_5,x_4), Q_O=(x_1,x_2,x_6,x_4)\right\}, \\
h_3&=\left\{Q_I=(x_1,x_2,x_6,x_4), Q_O=(x_1,x_2,x_6,x_7)\right\}, \\
h_4&=\left\{Q_I=(x_1,x_2,x_6,x_7), Q_O=(x_1,x_2,x_6,x_8)\right\}.
\end{align*}

Starting with $h_1$, we see that the terms $x_1, 2\cdot x_2, 7\cdot(x_1\oplus x_2)$ are computable, so we partition them into blocks satisfying $4 - \rk(A) \leq 4 - |A|$, giving $P=\left\{\{x_1, 2\cdot x_2\}, \{7\cdot(x_1\oplus x_2)\}\right\}.$ As neither partition will become uncomputable after $h_1$, we simply apply the first Hadamard gate:

\vspace{1em}
\centerline{\includegraphics[scale=0.2]{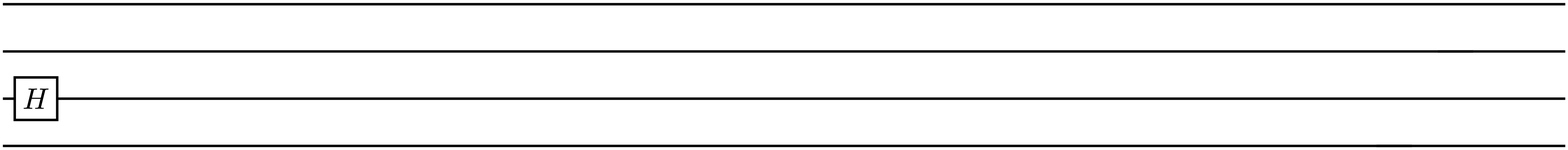}}
\vspace{1em}

At the second Hadamard gate, the path variable $x_5$ is available, so $x_5, 7\cdot(x_1\oplus x_5), 7\cdot(x_2\oplus x_5), x_1\oplus x_2\oplus x_5$ are now computable. We add them to the partition to get $$P=\left\{\{x_1, 2\cdot x_2, 7\cdot(x_2\oplus x_5)\}, \{7\cdot(x_1\oplus x_2)\}\{x_1\oplus x_2\oplus x_5, 7\cdot(x_1\oplus x_5), x_5\}\right\}.$$ We trivially see that $x_2\oplus x_5\notin\spn(\{x_1,x_2,x_6,x_4\})$ and likewise neither is $x_5$, so we synthesize a circuit computing the partitions $\{x_1, 2\cdot x_2, 7\cdot(x_2\oplus x_5)\}$ and $\{x_1\oplus x_2\oplus x_5, 7\cdot(x_1\oplus x_5), x_5\}$ and allow $\{7\cdot(x_1\oplus x_2)\}$ to move past the second Hadamard gate.

\vspace{1em}
\centerline{\includegraphics[scale=0.2]{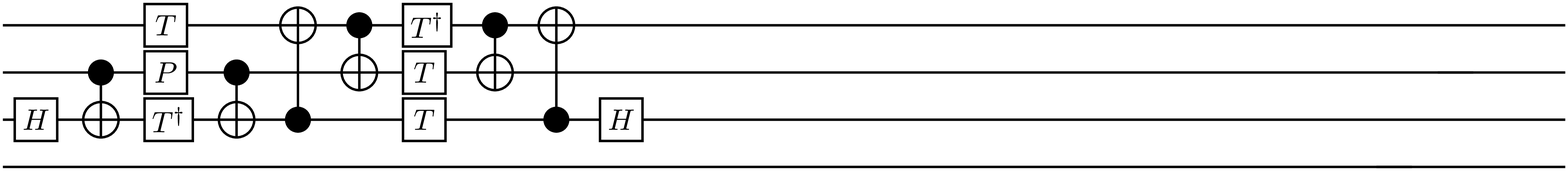}}
\vspace{1em}

Again, $x_6$ is now available, so insert the newly computable terms $x_6, 7\cdot(x_2\oplus x_6)$ into the current partition $P=\left\{\{7\cdot(x_1\oplus x_2)\}\right\}$ to get $P=\left\{\{7\cdot(x_1\oplus x_2), x_6, 7\cdot(x_2\oplus x_6)\}\right\}$. As the single partition block will still be computable after applying $h_3$, we apply the next Hadamard gate:

\vspace{1em}
\centerline{\includegraphics[scale=0.2]{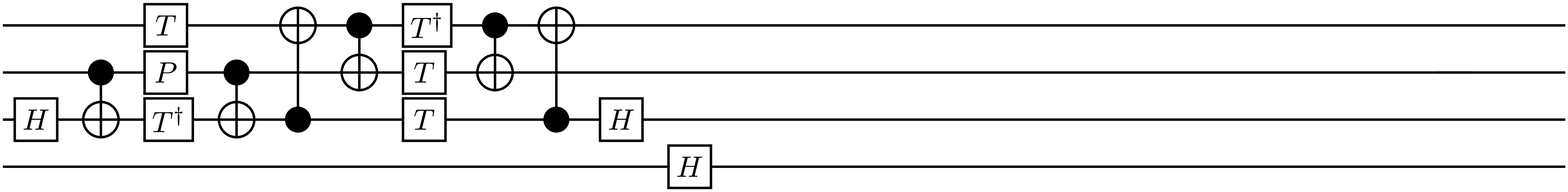}}
\vspace{1em}

We now add the last of the terms to the partition, $x_7, 7\cdot(x_6\oplus x_7), x_2\oplus x_6\oplus x_7$, giving

$$\left\{\{7\cdot(x_1\oplus x_2), x_6, 7\cdot(x_2\oplus x_6), x_7\},\{7\cdot x_2\oplus x_7), 7\cdot(x_6\oplus x_7), x_2\oplus x_6\oplus x_7\}\right\}.$$

As both partitions contain the value $x_7$ which will be destroyed by $h_4$, we apply the remaining partitions, followed by the last Hadamard gate:

\vspace{1em}
\centerline{\includegraphics[scale=0.2]{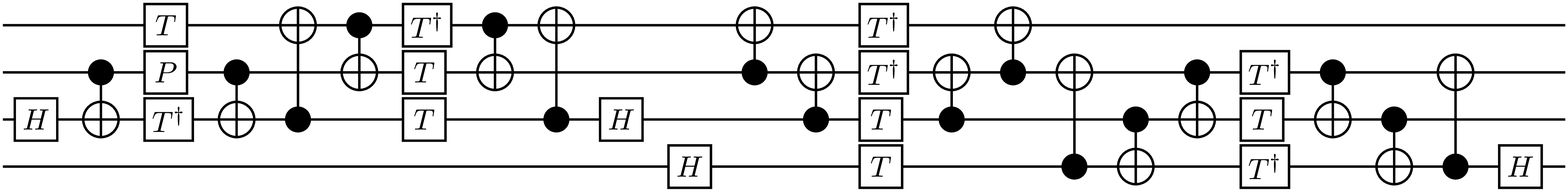}}
\vspace{1em}

The final circuit, shown above, reduces the original circuit by 2 $T$ gates (from 14 to 12) and 2 levels of $T$-depth (from 6 to 4).

Note that the partitions used in the above are minimal, but are not the partitions Algorithm~\ref{alg:partition} actually produces.  Instead, these partitions have been created to best demonstrate the algorithm. Additionally, for simplicity, we described states without parity, as there are no bit flip gates in this example.

\section*{Acknowledgments}
We would like to thank Martin R\"{o}tteler and Bill Cunningham for many helpful discussions.

Supported in part by the Intelligence Advanced Research Projects Activity (IARPA) via Department of Interior National Business Center Contract number DllPC20l66. The U.S. Government is authorized to reproduce and distribute reprints for Governmental purposes notwithstanding any copyright annotation thereon. Disclaimer: The views and conclusions contained herein are those of the authors and should not be interpreted as necessarily representing the official policies or endorsements, either expressed or implied, of IARPA, DoI/NBC or the U.S. Government. 

This material is based upon work partially supported by the National Science Foundation (NSF), during D. Maslov's assignment at the Foundation.
Any opinion, findings, and conclusions or recommendations expressed in this material are those of the author(s) and do not necessarily reflect the views of the National Science Foundation.

Michele Mosca is also supported by Canada's NSERC, MPrime, CIFAR, and CFI. IQC and Perimeter Institute are supported in part by the Government of Canada and the Province of Ontario.

\bibliographystyle{hieeetr}

\end{document}